\documentclass{bioinfo}
\pdfoutput=1
\usepackage{tikz}
\usepackage{amsthm}
\usepackage{pgfplots}
\usepackage{caption}
\usepackage{subcaption}
\usepackage{amsmath}
\usepackage{amssymb}
\usepackage{graphicx}
\usepackage{algorithm}
\usepackage{setspace}
\usepackage{algpseudocode}
\usepackage{soul}
\usepackage{xspace}
\usepackage{comment}
\usepackage{cleveref}
\usepackage{epstopdf}
\copyrightyear{2015} \pubyear{2015}
\access{Advance Access Publication Date: Day Month Year}
\appnotes{Manuscript Category}

\usetikzlibrary{snakes,arrows,shapes}
\DeclareMathOperator{\E}{\mathbb{E}}

\newcommand{\kmer}{$k$-mer\xspace}
\newcommand{\kmers}{$k$-mers\xspace}
\newcommand{\kpomer}{$(k+1)$-mer\xspace}
\newcommand{\kpomers}{$(k+1)$-mers\xspace}
\newcommand{\algname}{{\sc TwoPaCo}\xspace}

\newcommand{\ecoli}{{\it E.coli}\xspace}
\newcommand{\bwtbased}{{\tt bwt-based}\xspace}
\newtheorem{obs}{Observation}

\tikzstyle{every picture}+=[font=\sffamily]
\begin{document}
\raggedbottom
\algnotext{EndFor}
\algnotext{EndIf}
\algnotext{EndWhile}
\renewcommand{\algorithmicrequire}{\textbf{Input: }}
\renewcommand{\algorithmicensure}{\textbf{Output: }}

\firstpage{1}

\subtitle{Subject Section}

\title[TwoPaCo]{TwoPaCo: An efficient  algorithm to build the compacted de Bruijn graph from many complete genomes}
\author[Minkin \textit{et~al}.]{Ilia Minkin$^{1}$, Son Pham$^{2}$, Paul Medvedev$^{1,3,4,*}$}
\address{
$^1$Department of Computer Science and Engineering, The Pennsylvania State University, USA \\ 
$^2$Salk Institute for Biological Studies, USA \\
$^3$Department of Biochemistry and Molecular Biology, The Pennsylvania State University, USA\\
$^4$Genomic Sciences Institute of the Huck,  The Pennsylvania State University, USA
}

\corresp{$^\ast$To whom correspondence should be addressed.}

\history{Received on XXXXX; revised on XXXXX; accepted on XXXXX}

\editor{Associate Editor: XXXXXXX}

\abstract{\textbf{Motivation:}	
De Bruijn graphs have been proposed as a data structure to facilitate the analysis of related whole genome sequences,
in both a population and comparative genomic settings.
However, current approaches do not scale well to many genomes of large size (such as mammalian genomes).\\
\textbf{Results:}	In this paper, we present \algname, a simple and scalable low memory algorithm
	for the direct construction of the compacted de Bruijn graph from a set of complete genomes.
	We demonstrate that
	it can construct the graph for 100 simulated human genomes in less then a day and 
	eight real primates in less than two hours, on a typical shared-memory machine.
	We believe that this progress will enable novel biological analyses of hundreds of mammalian-sized genomes. \\ 
\textbf{Availability:}	Our code and data is available for download from {\tt github.com/medvedevgroup/TwoPaCo} \\
\textbf{Contact:} \href{ium125@psu.edu}{ium125@psu.edu} 
}

\maketitle

\section{Introduction}

The study of related features across different genomes is fundamental to many areas of biology,
both for genomes of the same species (pan-genome analysis) and for genomes across different species (comparative genomics).
The starting point in these studies is some representation of the relationship between genomes, often as a multiple alignment~\citep{gusfieldbook} or as a graph representation~\citep{lee02multiple}.
With the ubiquity of cheap sequencing,
the number of genome sequences available for these studies has expanded tremendously~\citep{genome10k,avian,genome10kb}.
The type of genomes available has also expanded: we have whole genomes, as opposed to only genic sequences,
and we now have many mammalian sized ($\sim3$ Gbp) genomes.
In addition, novel long-read sequencing technologies like Oxford Nanopore promise to make such genomes even easier to obtain.
Thus, we expect to have hundreds of whole mammalian genome sequences for comparison, in both the population and comparative genomic settings.
However, our current computational ability to analyze such large datasets is, at best, limited.

A major bottleneck toward the goal of comparing hundreds of whole mammalian genomes are 
scalability issues due to the problem of repeats.
Multiple alignment is a computationally hard problem due to the presence
of high copy-count repeats, which are absent in many lower-order species but 
cover roughly half of a mammalian genome.
For example, the human genome contains over a million ALU repeats.
Most multiple alignment methods mask repeats due to the computational challenge of handling them,
resulting in a loss of important features.
Without masking repeats, most approaches do not scale well to modern data, 
both in terms of computation time and memory usage.
A competition of whole-genome aligners demonstrated 
that some recent  tools are able to handle larger data sets;
however, these were still limited to $\leq20$ genomes of length $<200$ Mbp~\citep{alignathon}.

As an alternative to multiple alignment,
de Bruijn graph 
approaches for comparing whole genome sequences have been proposed~\citep{raphael2004novel,pham2010drimm,minkin2013sibelia}.
De Bruijn graphs have traditionally been used for {\em de novo} assembly~\citep{miller_assembly_2010,schatz2010assembly}, but in the case of already assembled genomes, they are built from a few long sequences, as opposed to billions of short reads.
In the setting of population genomics, a de Bruijn graph representation 
of closely related genomes can be used to discover polymorphism in a population~\citep{iqbal2012novo,dilthey15}.
The use of graphs brings up a host of other important problems that have been studied:
which data structure to use~\citep{pancake,dilthey15},
how to design efficient querying indices~\citep{siren14,holley2015bloom},
and how to do align read data to such graphs~\citep{huang13,paten14mapping}.
In the comparative genomics setting, a de Bruijn graph representation can be used to detect synteny blocks across different species~\citep{pham2010drimm,minkin2013sibelia}.

The construction of a de Bruijn graph is one of the most resource intensive steps of many of these algorithms and 
thus poses the major scalability bottleneck.
Recent papers have demonstrated how to efficiently construct the graph 
in the whole genome sequence setting~\citep{minkin2013sibelia,marcus2014splitmem,cazaux2014indexing,beller2015efficient,baier15}.
The fastest algorithm to date
was able to process seven whole mammalian genomes in under eight hours~\citep{baier15}.
However, constructing the graph is still prohibitive for larger inputs.

In this paper, we present \algname, a novel algorithm for constructing de Bruijn graphs from whole genome sequences.
We demonstrate that
it can construct the graph for 100 human genomes in less then a day and 
eight primates in less than two hours, on a typical shared-memory machine.
\algname is based on the following key insight.
We start with a basic naive algorithm, which has a prohibitively large memory usage but 
has the benefit that it is easily parallelizable.
We then create a two pass algorithm that uses the naive one as a subroutine. 
In the first pass, we use a probabilistic data structure to drastically reduce the
size of the problem, and in the second pass, 
we run the naive algorithm on the reduced problem.
One of our key design principles was to make the algorithm simple and embarrassingly parallelizable,
in order to take advantage of multi-thread support of most shared-memory servers.
The result is a simple and scalable low memory algorithm for the direct construction of the compacted de Bruijn graph for a set of complete genomes.

\section{Preliminaries}

For a string $x$, we denote by $x[i .. j]$ the substring from positions $i$ to $j$, inclusive of the endpoints.
We say that a string $x$ is the {\em prefix} of a string $y$, if $x$ constitutes the first $|x|$ characters of $y$, where $|x|$ is the length of $x$.
A string $x$ is the {\em suffix} of a string $y$, if $x$ constitutes the last $|x|$ characters of $y$.
At first, we define the {\em de Bruijn graph} built from a single string.
For a string $s$ and an integer $k$, we designate the de Bruijn graph as $G(s, k)$.
Its vertex set consists of all substrings of $s$ of length $k$, called {\em $k$-mers}.
Two vertices $u$ and $v$ are connected with a directed edge $u \rightarrow v$ if $s$ contains a substring $e$, $|e| = k + 1$ such that $u$ is the prefix of $e$ and $v$ is the suffix of $e$.
We will use terms ``$k$-mer" and ``vertex" interchangeably, as well as ``$(k+1)$-mer" and ``edge."
For clarity of presentation, we have defined the de Bruijn graph as a simple graph, but we in fact store it as a multi-graph.

Now we define the de Bruijn graph for multiple strings.
The union of two graphs $G_1 = (V_1, E_1)$ and $G_2 = (V_2, E_2)$ is the graph $G_1 \cup G_2 = (V_1 \cup V_2, E_1 \cup E_2)$.
For a collection of strings $S = \{s_1, s_2, \ldots, s_n\}$ and an integer $k$, 
the de Bruijn graph is the union of the graphs constructed from individual strings, i.e. $G(S, k) = G(s_1, k) \cup G(s_2, k) \cup \ldots \cup G(s_n, k)$.
Fig. \ref{ordinary_de_Bruijn} shows an example of a graph built from two strings.
Recall that a {\em path} through a graph is a sequence of adjacent vertices where the only repeated vertices may be the first and last one, whereas a {\em walk} can repeat both vertices and edges.
We say that a walk or path $p$ in the de Bruijn graph $G(S, k)$ \textit{spells} a string $t$ if $G(t, k) = p$.
We say that a vertex $v$ is a {\em bifurcation} if at least one of the following holds (1) $v$ has more than one incoming edge (2) $v$  has more than one outgoing edge.
A vertex $v$ is a {\em sentinel} if it is a first or last $k$-mer of an input string.
We call a vertex a {\em junction} if it is a bifurcation, or a sentinel, or both.
The set $J(s, k)$ is the set of positions $i$ of the string $s$ such that the \kmer $s[i .. i+k-1]$ is a junction.
For a collection of strings $S$ the set $J(S, k)$ is defined analogously.

\begin{figure}[H]
	\centering
	\begin{subfigure}{.49\textwidth}
		\includegraphics[scale=.4]{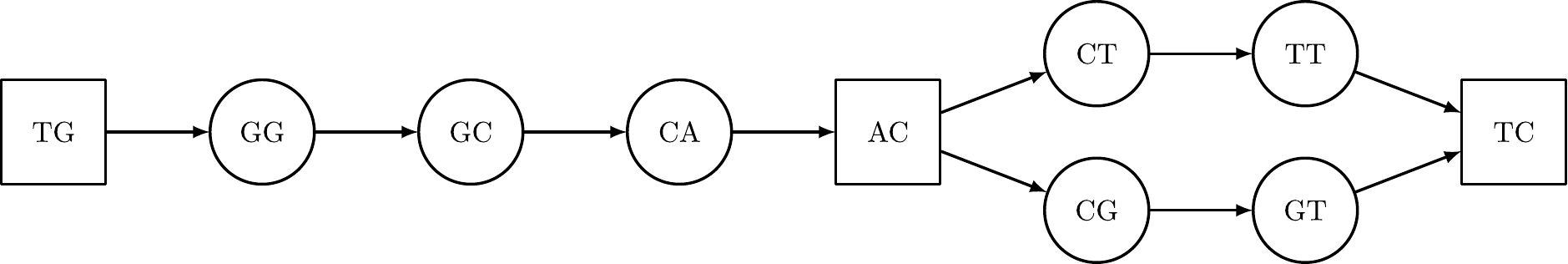}
		\caption{}
		\label{ordinary_de_Bruijn}
	\end{subfigure}
	\hfill

	\begin{subfigure}{.23\textwidth}
		\includegraphics[scale=.5]{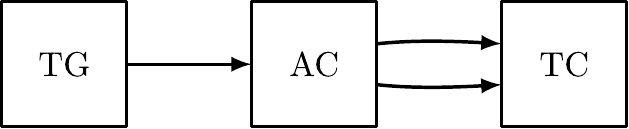}
		\caption{}
		\label{compacted_de_Bruijn}
	\end{subfigure}
	\hfill
	\begin{subfigure}{.23\textwidth}
		\includegraphics[scale=1.5]{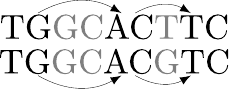}
		\caption{}
		\label{example_edge_set_construction}
	\end{subfigure}
	\hfill

	\begin{subfigure}{.37\textwidth}
		\includegraphics[scale=.5]{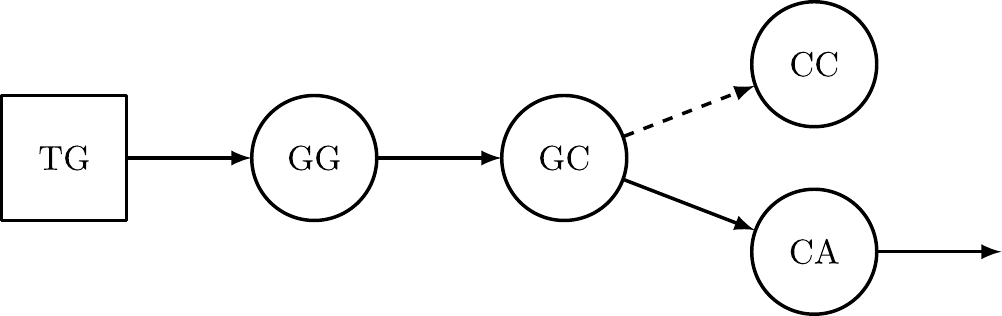}
		\caption{}
		\label{bloom_filter_false_edges}
	\end{subfigure}
	\caption{The de Bruijn graph and its compacted version. 
		(\subref{ordinary_de_Bruijn}) An example of an ordinary de Bruijn graph built from the genomes $S = \{``TGGCACGTC", ``TGGCACTTC"\}$ and $k = 2$.
		Junctions are indicated by square vertices.
		(\subref{compacted_de_Bruijn}) Graph obtained after compaction.
(\subref{example_edge_set_construction}) The two genomes that generate the graph, with the junction $k$-mers in bold; 
the arrows between them indicate edges in the compacted graph and non-branching paths in the ordinary graph.
The strings b	etween them label the edges in the compacted graph.
(\subref{bloom_filter_false_edges}) If we store edges in a Bloom filter, we may observe false edges 
(dotted line) in the ordinary graph; this can lead to detection of false junctions, like the vertex ``GC'' in this case.}
\end{figure}

A de Bruijn graph can be compacted by collapsing non-branching paths into single edges.
More precisely a {\em non-branching path} in an ordinary de Bruijn graph is a path $u \rightsquigarrow v$ such that the only junction vertices on this path are possibly $u$ or $v$.
The {\em compaction} of a non-branching path $p = u \rightsquigarrow v$ is removal of edges of $p$ and replacing it with an edge $u \rightarrow v$.
A {\em maximal} non-branching path is a non-branching path that cannot be extended by adding an edge.
The {\em compacted graph} $G_c(S, k)$ is the graph obtained from $G(S, k)$ by compaction of all its maximal non-branching paths.
This graph is sometimes referred to as the compressed graph in the literature~\citep{beller2015efficient}.
It is easy to see that the vertex set of $G_c(S, k)$ is the set of junctions of the graph $G(S, k)$ and two vertices $u$ and $v$ of $G_c(S, k)$ are connected if there is a non-branching path $u \rightsquigarrow v$  in $G(S, k)$.
Fig. \ref{compacted_de_Bruijn} shows an example of a compacted de Bruijn graph.
Note that a compacted graph is a multi-graph: after compaction a pair of vertices can be connected by edges going in the same direction that corresponded to different paths in the ordinary graph.

Graph compaction is the first step of most algorithms working with de Bruijn graphs, 
since it drastically reduces the number of vertices.
It can be obtained from the ordinary graph in linear time by a simple graph traversal. 
However, building and storing the ordinary graph takes lots of space, which 
we seek to avoid in our algorithm by constructing the compacted graph directly.

A {\em Bloom filter} is a space efficient data structure for representing sets that supports two operations: 
storing an element in the set and checking if an element is in the set~\citep{bloom1970space}.
A Bloom filter offers improvements in space usage but can generate false positives during membership queries. 
Bloom filters have previously been successfully applied 
to assembly~\citep{bfcounter,chikhi2013space,cascading,bless} and 
to indexing and compression of whole genomes~\citep{holley2015bloom}.

\section{Reduction to the Problem of Finding Junction Positions}\label{sec:reduction}
\algname is based on the observation that there is a bijection between maximal non-branching paths of the de Bruijn graph and substrings of the input 
whose junctions are exactly the two flanking \kmers (Observation~\ref{lem:junction} below). 
This observation reduces the problem of graph compaction to finding the set of junction positions $J(S, k)$, 
as follows.
The vertex set of the compacted graph is the set of all \kmers located at positions $J(S,k)$.
To construct the edges, we need to find substrings flanked by junctions.
To do this, we can traverse positions of $J(S,k)$ in the order they appear in the input.
For every two consecutive junction positions $i$ and $j$, we record an edge between the \kmer at $i$ and the \kmer at $j$.
\Cref{example_edge_set_construction} shows an example of how sequences of junctions generate non-branching paths in the ordinary graph and edges in the compacted one.

The observation follows in a straight-forward way from the definitions, 
but we state and prove it here for completeness.
\begin{obs}\label{lem:junction}
	Let $s$ be an input string and $P$ be the set of maximal non-branching paths of the graph $G(s, k)$.
	Let $T$ be the set of substrings of $s$ such that each $t \in T$ starts and ends with a  junction of $G(s, k)$ and does not contain junctions in between.
	Then there exists a bijective function $g : T \rightarrow P$.
\end{obs}
\begin{proof}
	Let $g$ be the function mapping substrings of $s$ to walks in $G(s,k)$,
	where $g$ maps a substring  to the vertices corresponding to its constituent \kmers.
	To prove that $g$ is a bijection when restricted to $T$, 
	we have to show that it is both an injection and surjection.
	Note that $g$ is injective by construction, that is, 
	any walk is spelled by a unique string.
	To prove that it is surjective, we need to show that for any maximal non-branching path 
	$p = u \rightsquigarrow v$, there is a $t\in T$ such that $g(t) = p$. 
	That is, $p$ is spelled by a string in $T$.
	Since the walk $g(s)$ must traverse all vertices in the graph, and the internal
	vertices of $p$ have in- and out- degrees equal to one, 
	the walk $g(s)$ must contain $p$ as a subwalk.
	Hence, the string $t$ spelled by $p$ must be a substring of $s$, i.e. $g(t) = p$.
	The internal \kmers of $t$ are non-junctions because $p$ is non-branching,
	and the first and last \kmers of $t$ are junctions because $p$ is maximal.
	Hence, $t \in T$.
\qed
\end{proof}
Generalization of the observation to the case of multiple strings is straightforward.

\begin{methods}

\section{Single Round Algorithm}

\begin{algorithm}[h]
\caption{\textit{Filter-Junctions}} 
\label{vertex_set_algorithm}
\algorithmicrequire{strings $S = \{s_1, \ldots, s_n\}$, integer $k$, 
and an empty set data structure $E$. 
A candidate set of marked junction positions $C \supseteq J(S, k)$ is also given.
When the algorithm is run naively, all the positions would be marked. 
}\\
\algorithmicensure{a reduced candidate set of junction positions.}
\begin{algorithmic}[1]
\For{$s \in S$} \label{line:edge:start}
	\For{$1\leq i < |s| - k$} 
		\If{$C[s,i] = $ marked}  \Comment Insert the two \kpomers containing the \kmer at $i$ into $E$.
			\State Insert $s[i ..  i + k]$ into $E$.
			\State Insert $s[i-1 .. i - 1 + k]$ into $E$.
		\EndIf
	\EndFor
\EndFor \label{line:edge:end}
\For{$s \in S$}	\label{line:check:start}
	\For{$1\leq i < |s| - k$} 
		\If{$C[s,i] = $ marked \textbf{and} $s[i ..  i + k - 1]$ is not a sentinel}
			\State $in \gets 0$ \label{line:count:start}		\Comment Number of entering edges
			\State $out \gets 0$ 		\Comment Number of leaving edges
			\For{$c \in \{A, C, G, T\}$}	\Comment Consider possible edges and count how many of them exist
				\If{$v \cdot c \in E$} \label{line:check1}	\Comment The symbol $\cdot$ depicts string concatenation
					\State $out \gets out + 1$		
				\EndIf
				\If{$c \cdot v \in E$} \label{line:check2}
					\State $in \gets in + 1$
				\EndIf
				\EndFor \label{line:count:end}
			\If{$in = 1$ \text{\textbf{and}} $out = 1$} \Comment If the \kmer at $i$ is not a junction.
			\State $C[s,i] \gets$ Unmarked \label{line:unmark}
			\EndIf
		\EndIf
	\EndFor
\EndFor \label{line:check:end}
\State \Return $C$
\end{algorithmic}
\end{algorithm}
In the previous section, we reduced the problem of constructing a compacted de Bruijn graph to that of finding the locations in the genome where junction vertices are located.
We will now present our algorithm for finding junction positions, in increasing layers of complexity.
First, we will describe \Cref{vertex_set_algorithm}, which can already be used as a naive algorithm to identify the junctions.
However, \Cref{vertex_set_algorithm} alone has a prohibitively large memory footprint.
To address this, we will present \Cref{alg:twopass}, which uses \Cref{vertex_set_algorithm} as a subroutine 
but reduces the memory requirements.
In cases of very large inputs, even \Cref{alg:twopass} can exceed the available memory. 
In \Cref{sec:split}, we finally present \Cref{construct_algorithm}, which addresses this limitation.
It limits memory usage, at the expense of time, by calling
\Cref{alg:twopass} over several rounds.
We refer to this final algorithm (\Cref{construct_algorithm}) as \algname.

In \Cref{vertex_set_algorithm}, we start with a candidate set $C$ of junction positions in the genomes.
A set of positions $C$ is called a {\em candidate set} if $C \supseteq J(S,k)$ and any two positions that 
start with the same \kmer can either both present or both absent from $C$.
$C$ is represented using boolean flags which mark every position of the genomes which is present in the set.
If \Cref{vertex_set_algorithm} is used naively, it would be called with every position marked; 
in general, however, 
we can use $C$ to capture the fact that the unmarked positions have been previously eliminated from consideration as junctions.

First, we store all edges of the ordinary de Bruijn graph in a set $E$.
We do this by a linear scan and for every \kpomer at position $i$, 
if either of the \kmers at positions $i$ or $i+1$ are marked, we insert the \kpomer into the set $E$
(\Crefrange{line:edge:start}{line:edge:end}).
Second, we again scan through the genomes and consider the \kmer $v$ at every marked position.
We use $E$ to check how many edges in $G(S,k)$ 
enter and leave $v$ (\Crefrange{line:count:start}{line:count:end}).
Since the DNA alphabet is finite, we can do this by merely considering all eight possible \kpomers -- 
four entering, and four leaving -- and checking whether they are in $E$.
If the in- and out-degrees do not satisfy the definition of a junction, we unmark position $i$; otherwise, we leave it marked.

\Cref{vertex_set_algorithm} can be used naively to find all junction positions,
by initially marking every position as a potential junction.
Storing the set $E$ in memory, however, is infeasible for large datasets.
To reduce the space requirements, we develop the two pass~\Cref{alg:twopass}.
In the first pass, we run~\Cref{vertex_set_algorithm}, 
but use a Bloom filter to store the set $E$ instead of a hash table.
A Bloom filter takes significantly less space than a hash table; however, the downside is that it can generate false positives during membership queries.
That is, when we check if a \kpomer is present in $E$ (\Cref{line:check1,line:check2} in \Cref{vertex_set_algorithm})
we may receive an answer that it is present, when it is in reality absent.
The effect is that the calculated in- and out-degrees may be inflated and we may leave non-junctions marked (\Cref{line:unmark}), see Fig. \ref{bloom_filter_false_edges}.
Nevertheless, the marked positions still represent a candidate set of junctions, since a junction will never be unmarked.
Thus, running \Cref{vertex_set_algorithm} with the Bloom filter reduces memory 
but does not always unmark non-junction positions.
In order to eliminate these marks, 
we run~\Cref{vertex_set_algorithm} again, using the positions marked in the first pass as a starting point, 
but this time using a hash table to store $E$ (\Cref{line:secondpass} in \Cref{alg:twopass}).
This second pass will unmark all remaining marked non-junction positions. 
Since the set of candidate marks has been significantly reduced after the first pass, 
the memory use of the hash table is no longer prohibitive.
As with \Cref{vertex_set_algorithm}, \Cref{alg:twopass} can be used to find all junction positions 
by initially marking every position as a potential junction.

\begin{algorithm}[h]
\caption{\textit{Filter-Junctions-Two-Pass}} 
\label{alg:twopass}
\algorithmicrequire{strings $S = \{s_1, \ldots, s_n\}$, integer $k$, 
a candidate set of junction positions $C_\text{in}$,
integer $b$} \\ 
\algorithmicensure{a candidate set of junction positions $C_\text{out}$
}
\begin{algorithmic}[1]
		\State $F \gets $ an empty Bloom filter of size $b$
		\State $C_\text{temp} \gets \textit{Filter-Junctions}(S, k, F, C_\text{in})$ 	\label{line:firstpass}		\Comment The first pass
		\State $H \gets $ an empty hash table
		\State $C_\text{out} \gets \textit{Filter-Junctions}(S, k, H, C_\text{temp})$	\label{line:secondpass}		\Comment The second pass
		\State \Return $C_\text{out}$
\end{algorithmic}
\end{algorithm}

Our implemented algorithms also handle the reverse complementarity of DNA, 
using standard techniques.
We summarize this briefly for the sake of completeness.
For a string $s$, let $\bar{s}$ be its reverse complement,
and define the comprehensive de Bruijn Graph as the graph $G_\text{comp}(s, k) = G(s, k) \cup G(\bar{s}, k)$; 
the graph for multiple strings and the compacted graph is defined analogously.
To build the compacted comprehensive graph, we have to modify~\Cref{vertex_set_algorithm} so that $E$ represents
each \kmer and its reverse complement jointly. 
For example, this can be done by always storing the {\em normalized} form of a \kmer, 
which is the lexicographically smallest string between the \kmer and its reverse complement~\citep{chikhi2014representation}.
Similarly, we have to be careful when we make membership queries to $E$ in~\Cref{vertex_set_algorithm}, 
so that we are always querying normalized \kmers.

\section{Multiple Rounds: Dealing with Memory Restrictions}\label{sec:split}
While~\Cref{alg:twopass} significantly reduces the memory usage, 
it is still possible that the hash table in the second pass may not fit into the main memory, 
for some very large inputs.
To deal with this issue, we develop~\Cref{construct_algorithm}, which splits the input \kmers into $\ell$ parts and runs~\Cref{alg:twopass} in $\ell$ rounds.
Each round processes only one part, thus limiting its memory use to what is available. 
We note that we must partition the \kmers, which is distinctly different from partitioning the positions.
In particular, if two different positions have the same \kmer, they must belong to the same class;
hence, we cannot simply divide our strings into chunks. 
When $\ell=1$, \Cref{construct_algorithm} reduces to \Cref{alg:twopass} and does not limit its memory use,
but when $\ell$ is increased, the peak memory usage decreases at the expense of more rounds and hence longer running time. 

In each round, \Cref{construct_algorithm} will consider only approximately $1/\ell$ of the \kmers 
to check if they are junctions.
First, we partition the set of \kmers into $\ell$ classes (\Cref{line:partition}).
In round $i$, our algorithm begins by marking the positions whose \kmers are in class $i$ (\Cref{line:initc}).
Note that each position is considered in exactly one round.
We then call~\Cref{alg:twopass}, which unmarks those positions which are not junctions.
After all the rounds are complete, the junction vertices are exactly those that remain marked (\Cref{line:finalc}).

The maximum memory usage of~\Cref{construct_algorithm} is minimized when the partition created in~\Cref{line:partition} 
leads to an equally sized hash table in every round.
The hash table at round $i$ stores the set of \kpomers that contain a \kmer from partition $i$, which we denote $E_i(S,k)$.
Thus, we would like the sizes of $E_i(S,k)$ to be as equal as possible.
We are not concerned with obtaining an optimal partition, since a small discrepancy in the memory in each round is permissible.
We therefore develop the following heuristic.
Suppose that we have a hash function $f$ with range $[0, q)$, for some $q \gg \ell$.
We assign a counter $e_i$, for $i\in[0,q)$, to calculate an approximate value for $|E_i(S, k)|$, as follows.
We make a pass through the input and use a Bloom filter to store all the \kpomers.
Additionally, for every \kpomer, if it is already present in the Bloom filter, 
we increase the corresponding counters.
This way, we try to count only unique \kpomers, though the count can be slightly inflated by false positives.

Once we obtain the counters $e_i$, we amalgamate sets $E_i(S, k)$ into $\ell$ ones.
This problem is equivalent to the number partitioning problem, which is NP-hard \citep{garey1979computers}, 
so we use a greedy heuristic based on the linear scan of numbers $e_i$.
According to this heuristic, the first class $E_1(S, k)$ corresponds to the first $t$ subranges such that $\sum_{1 \leq i \leq t} e_i \leq \sum e_j / \ell $, and $t$ is as large as possible.
Other classes are determined analogously.

\begin{algorithm}[h]
\caption{\algname}
\label{construct_algorithm}
\algorithmicrequire{strings $S = \{s_1, \ldots, s_n\}$, integer $k$, integer $\ell$, integer $b$} \\ 
\algorithmicensure{the compacted de Bruijn graph $G_c(S, k)$ }
\begin{algorithmic}[1]
	\State $C_{\text{init}}\gets $ boolean array with every position unmarked
	\State Divide \kmers of $S$ into $\ell$ partitions. \label{line:partition}
	\For{$0 \leq i < \ell $}
		\State $C_i \gets $ mark every position of $C_{\text{init}}$ which belongs to partition $i$. \label{line:initc}
		\State $C'_i \gets \text{Filter-Junctions-Two-Pass}(S, k, b, C_i)$
\EndFor
\State $C_{\text{final}} = \bigcup C'_i$ \label{line:finalc}
\State \Return Graph implied by $C_{\text{final}}$, as described in~\Cref{sec:reduction}.
\end{algorithmic}
\end{algorithm}

\section{Parallelization Scheme}\label{sec:par}
We designed our algorithm so that it can be effectively parallelized on a multi-processor shared memory machine.
The bulk of the computation happens in Algorithm \ref{vertex_set_algorithm}, which consists of two parts.
Each part is a loop over all the positions in the input, \Crefrange{line:edge:start}{line:edge:end} in the first part and \Crefrange{line:check:start}{line:check:end} in the second.
The first loop is embarrassingly parallelizable as long as the data structure representing the set $E$ supports concurrent writes.
We use a lock-free Bloom filter when \Cref{vertex_set_algorithm} is called during the first pass of \Cref{alg:twopass}, and a concurrent hash table when it is called during the second pass.
The second loop is trivially parallelizable: threads will get non-overlapping portion of genomes, hence the synchronization on $C$ is not needed.
A synchronization barrier separates the two loops.
The compacted edge generation step that we discussed in the Section \ref{sec:reduction} is embarrassingly parallelizable as well.

We implement the parallelization using the standard single producer/multiple consumer pattern \citep{oaks2004java}.
According to this design pattern we create
(1) a single reader thread that splits the input into equal sized substrings and puts them into worker queues, and
(2) many worker threads that dequeue and process the substrings.
We utilized parallel programming primitives from the Intel's Threading Building Blocks library \citep{reinders2007intel}.
Note that this way we store only part of the input and the corresponding array $C$ in the input to save memory.

\section{Theoretical Analysis and Comparison} \label{sec:analysis}

In this section, we will analyze the running time and memory usage of our algorithm, and compare it with that of other algorithms.
Suppose that the de Bruijn graph $G(S, k)$ has $E$ edges, $J$ junctions and $L$ non-junctions that we call \textit{links}.
First, we will analyze the number of false positive junctions. 
A false positive junction is a link whose positions in $S$ are incorrectly left marked at the end of the first pass.
We assign an indicator variable $I_\ell$ to each link $\ell$, $I_\ell = 1$ if the link $\ell$ is
a false positive junction and $I_\ell = 0$ otherwise.
This way, the total number of false positive junctions is $FP = \sum_{1 \leq \ell \leq L} I_\ell$.
Let the probability that a link is a false positive junction be $p$.
By linearity of expectation we have $\E[FP] = \E[\sum_{1 \leq \ell \leq L} I_\ell] = L p$.
To calculate the probability $p$, note that each link has exactly one incoming and one outgoing true edge.
Hence, querying the Bloom filter in \Cref{line:check1} and \Cref{line:check2} of the \Cref{vertex_set_algorithm} may discover at most six false edges: three incoming and three outgoing ones.
At least one false positive from those six queries results in the link misclassified as a junction.
\citet{mitzenmacher2005probability} show that the probability of a single false 
positive resulting from querying a Bloom filter is 
$q = (1 - e ^ {-h E / b})^h$.
where $h$ is the number of hash functions used by the Bloom filter and $b$ is the number of bits in the filter.
Assuming that queries are independent, $p = 1 - (1 - q)^6 = 1 - (1 - (1 - e ^ {h E / b})^h)^6$.

Now we will analyze the running time.
Let $m$ be the total length of the input strings.
First, note that storing and querying \kmers with the Bloom filter requires calculation of $h$ hash values for each operation.
We use a family of sliding window hash functions, so both filling and querying the Bloom filter in the first pass takes $O(m h)$ operations.
In the second pass the algorithm employs a hash table to store and query \kpomers.
Denote by $M$ the number of marks left in the array $C$ after the first pass.
The expected running time is then $O(m h + M k)$, since each hash table operation takes $k$ time and 
there are $O(M)$ operations total.
To calculate $M$, let's assume that the average number of times a false positive junction occurs in the input string is given by $r$. 
Then, the expected value of $M$ is $|G_c| + L p r$, 'where $|G_c|$ is the number of edges in the compacted de Bruijn multi-graph.
The expected running time is then $O(m h + (|G_c| + L p r) k)$

To calculate the memory usage, note that the first pass allocates $b$ bits of memory for the Bloom filter and the second pass uses a hash table that contains at most $8(J + FP)$ elements.
Hence, the expected memory usage is $O(\max[b, (J + L  p) k])$.
The array $C$ of marks is accessed sequentially by the algorithm and can be stored in the external memory without loss of performance.
As discussed in \Cref{sec:par}, at each moment the memory contains only a constant amount of characters of the input strings, so the input length does not contribute to the asymptotic bound.

\begin{table*}[t]
	\caption{
		\small
		Running times and memory consumption of different algorithms for constructing 
		the de Bruijn graph from multiple complete genomes.
		For SplitMEM $g$ stands for the size of the largest genome in the input.
		An explanation of other variables is given in the \Cref{sec:analysis}.
		\label{table:analysis_comparison}
	}
	\small
	\begin{center}
		\setlength\tabcolsep{5pt} 
		\begin{tabular}{|l|r|r|}\hline
			Algorithm & Running Time & Memory \\ \hline 
			Sibelia~\citep{minkin2013sibelia} & $O(m)$ & $O(m)$  \\
			SplitMEM~\citep{marcus2014splitmem} & $O(m \log g)$ & $O(m + |G_c|)$ \\
			{\tt bwt-based} from \citet{baier15} & $O(m)$ & $O(m)$ \\
			\algname &$O(m h + (|G_c| + L p r) k)$ & $O(\max[b, (J + L  p) k])$ \\ \hline
		\end{tabular}
	\end{center}
\end{table*}

\Cref{table:analysis_comparison} contains asymptotic upper bounds on memory usage and running times of different algorithms for constructing the compressed de Bruijn graph from multiple complete genomes.
The performance of \algname depends highly on the number of junctions present.
On practical instances of related genomes datasets, there is a lot of shared sequence and the number of junctions is low.
Unlike other algorithms, our expected memory usage depends only on the structure of the input, but not directly on its size.
At the same time, dependence on $k$ makes \algname less applicable in case of very large $k$.
\end{methods}
\section{Results}
To evaluate the performance of \algname, we conducted several experiments.
We compared its running time and memory footprint with other available implementations of de Bruijn graph compaction algorithms.
We then ran \algname on a real dataset of biological interest as well as a large dataset of simulated data.
We assessed the parallel scalability of our implementation and capabilities of running the algorithm on machines with limited memory using the round splitting procedure
Finally, we evaluated the effects of input length and structure on the running time and memory usage.

First, we benchmarked \algname against Sibelia \citep{minkin2013sibelia}, SplitMEM \citep{marcus2014splitmem} and 
the {\tt bwt-based} algorithm of \citet{baier15}, using default parameters.
As far as we understood, the algorithm in~\citet{beller2015efficient} was subsumed by~\citet{baier15}.
There were two important caveats.
First, in most genomics application, it is necessary to account for both strands in the de Bruijn graph.
To make SplitMem and \bwtbased work with both strands,
we appended the reverse complements of the sequences to the input, as suggested by their authors.
In our results, we show SplitMEM and the \bwtbased in two versions: (1) considering only one strand, and
(2) considering both strands.
Second, Sibelia not only constructs the compacted graph but also modifies it after construction.
We therefore ran Sibelia only in the construction mode (contrary to the bechmarks in~\citet{marcus2014splitmem}).

For benchmarking purposes, we used a dataset of 62 \ecoli genomes (310 Mbp) from~\citet{marcus2014splitmem}.
We also used a dataset with seven human genomes ($\sim$21 Gbp) used by~\citet{baier15}, 
which includes five different assemblies of the human reference genome and two paternal haplotypes of NA12878 (see~\citet{baier15} for more details).
We ran our experiments on the highest memory Amazon EC2 instance (r3.8xlarge):
a server with Intel Xeon E5-2670 processors and 244 GB of RAM.
We set the default number of internal hash functions in the Bloom filters to four.
We also verified the correctness of \algname by comparing its output to that of a naive compaction algorithm on feasible test cases.
A direct comparison to the output of other tools is impractical since each algorithm handles edges cases differently (e.g. the presence of undetermined nucleotides (Ns) in the input).

The results are shown in Table \ref{table:exp_results}.
For seven human genomes, \algname was 12 -- 14 times faster than \bwtbased with a single strand, when we used 15 threads.
When only a single thread was used, \algname was 1.8 -- 2.0 times faster.
When \bwtbased was run with both strands, our improvements were approximately doubled.

\begin{figure}[H]
\centering
\begin{minipage}{.45\textwidth}
\begin{tikzpicture}
\begin{axis}
[
	axis equal,
	title=Parallel scalability,
	xlabel=Number of worker threads,
	ylabel=Speedup (times),
	title=Parallel scalability,
	xlabel=Number of worker threads,
	ylabel=Speedup (times),
	legend entries={First pass, Second pass, Edge construction},
	legend style={at={(0.3,0.7)},anchor=south},
	width=1.\textwidth
]
	\addplot table [x=NTHREADS, y=ROUND1_ACC, col sep=comma] {results_scalability.csv};
	\addplot table [x=NTHREADS, y=ROUND2_ACC, col sep=comma] {results_scalability.csv};
	\addplot table [x=NTHREADS, y=EDGES_ACC, col sep=comma] {results_scalability.csv};
\end{axis}
\end{tikzpicture}
\end{minipage}
\caption{\small Parallel speedup of the different parts of \algname.  Edge constructions refers to the conversion of junction positions to the compacted graph, as described in~\Cref{sec:reduction}.  
The Bloom filter was $8.58$ GB and used eight internal hash functions. We set $k = 25$.}
\label{plot:scale}
\end{figure}
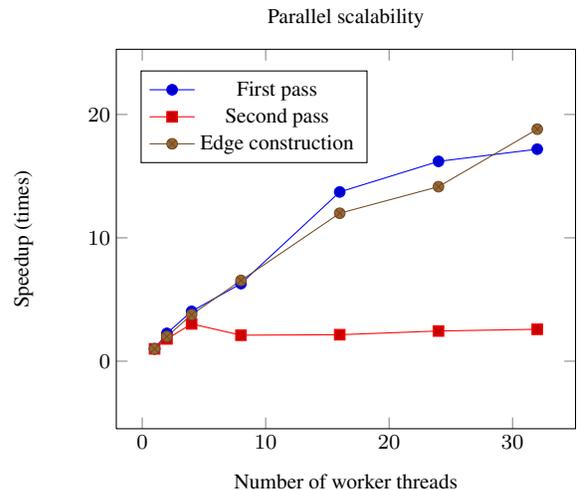

\begin{table*}[t]
	\caption{
		\small
		Benchmarking comparisons. 
		Each cell shows the running time in minutes and the memory usage in parenthesis in gigabytes.
		\algname was run using just one round, with a Bloom filter size $b=0.13$ GB for \ecoli and $4.3$ GB for human with $k =25$ and $b = 8.6$ GB with $k = 100$.
		A dash in the SplitMem column indicates that it ran out of memory,
		while a dash in the Sibelia column indicates that it could not be run on such large inputs.
		\label{table:exp_results}
	}
	\small
	\begin{center}
		\bgroup \def\arraystretch{1.3} 
		\setlength\tabcolsep{5pt} 

		\begin{tabular}{|l|r|r|r|r|r|r|r|}\hline
			& Sibelia~\citep{minkin2013sibelia} & 
			SplitMem~\citep{marcus2014splitmem} & 
			\multicolumn{2}{|c|}{\bwtbased from \citet{baier15}} & 
			\multicolumn{2}{|c|}{\algname} \\
			\hline
			& & single strand & single strand & both strands & 1 thread & 15 threads \\
			\hline 

			\ecoli ($k=25$) & 10 (12.2) & 70 (178.0) &  8 (0.85) & 12 (1.7) & 4 (0.16) & 2 (0.39) \\
			\ecoli ($k=100$) & 8 (7.6) & 67 (178.0) &  8 (0.50) & 12 (1.0) & 4 (0.19) & 2 (0.39) \\
			7 humans ($k=25$) & - & - &  867 (100.30) & 1605 (209.88) & 436 (4.40) & 63 (4.84) \\
			7 humans ($k=100$) & - & - & 807 (46.02) & 1080 (92.26) & 317 (8.42)  & 57 (8.75) \\
			\hline
		\end{tabular}

		\egroup
	\end{center}
\end{table*}
\begin{table}
	\caption{\small Results of running \algname on large datasets.
		Each cell shows the running time in minutes and the memory usage in parenthesis in gigabytes.
		\algname was run with 15 threads and the Bloom filter size was $b=34$ GB for primates and 
		$b=69$ GB for humans.
		An empty cell indicates we did not perform the experiment.
		\label{table:large_exp_results}
	}
	\small
	\begin{center}
		\bgroup \def\arraystretch{1.3} 
		\setlength\tabcolsep{5pt} 

		\begin{tabular}{|l|r|r|r|}\hline
			Dataset & 1 thread & 15 threads \\
			\hline 
			8 primates ($k=25$) & 914 (34.36) & 111 (34.36) \\
			8 primates ($k=100$) & 756 (56.06) & 101 (61.68) \\
			(43+7) humans ($k=25$) &   & 705 (69.77) \\
			(43+7) humans ($k=100$) &  & 927 (70.21) \\
			(93+7) humans ($k=25$) &  & 1383 (77.42) \\
			\hline
		\end{tabular}
		\egroup
	\end{center}
\end{table}

\begin{table}
	\caption{\small Number of marks the array $C$: initially, after the first pass, and after the second pass of \Cref{alg:twopass}. \label{table:marks_results}}
	\small
	\begin{center}		
		\begin{tabular}{|l|r|r|r|}\hline
			Dataset & Total Positions & First Pass & Second Pass \\ \hline
			\ecoli ($k=25$) &  310,157,564 &  24,649,489 & 24,572,562 \\
			\ecoli ($k=100$) &  310,157,489 & 22,848,018 & 9,492,091 \\ 	
			7 humans ($k=25$) & 21,201,290,922  & 3,489,946,013 & 2,974,098,154\\
			7 humans ($k=100)$ &  21,201,290,847 & 1,374,287,870 & 188,224,214 \\
			8 primates ($k=25$) & 24,540,556,921 &  5,423,003,377 & 5,401,587,503 \\ 
			8 primates ($k=100$) & 24,540,556,846  & 1,174,160,336 & 502,441,107 \\
			\hline
		\end{tabular}
	\end{center}
\end{table}


We also assessed \algname's ability to handle (1) large numbers of long closely-related genomes, and (2) more divergent genomes.
To do so, we generated 93 human genomes using the FIGG genome simulator \citep{killcoyne2014figg} and ``{\tt normal}'' simulation parameters.
The FIGG genome simulator generates complete sequences based on a reference genome and variations' frequencies extracted from the datasets from projects like \cite{10002010map} and \cite{gibbs2003international}.
The mutations comprise single-nucleotide alterations as well as indels and structural variations of larger size.
We ran \algname on three datasets: (1) 43 simulated genomes plus the seven used in \Cref{table:exp_results}, (2) 93 simulated human genomes plus the seven, and (3) eight primate genomes from the UCSC genome browser:
gibbon, gorilla, orangutan, rhesus, baboon, chimp, bonobo, and human.
We also tried to run other tools on the datasets above, but they ran out of memory.
The results are shown in \Cref{table:large_exp_results}.
We construct the graph for 100 human genomes in 23 hours using 77 GB of RAM and 15 threads.
For eight primates, we used under two hours and 34-62 GB of RAM on 15 threads.

For the benchmarks and real datasets in the experiments above, 
we recorded the number of marks that the \Cref{alg:twopass} left in the array $C$ after each stage
(\Cref{table:marks_results}).
We did not record those numbers for the larger datasets due to 
the associated cost restrictions of re-running the larger experiments.


To measure the parallel scalability of \algname, 
we fixed a dataset consisting of five simulated human genomes.
\Cref{plot:scale} shows scaling results for 1-32 worker threads.
The first pass of \Cref{alg:twopass}, and the conversion of junction vertices to the graph 
(as described in~\Cref{sec:reduction}), 
scale almost linearly up to 16 threads.
The second pass does not scale past four worker threads,
due to what we believe is the limited parallel performance of the 
concurrent hash table, which we plan to improve in the future.

Next, we evaluated the performance of \algname under memory restrictions.
For each run, we set a different memory threshold and checked how many rounds 
were necessary so that \algname did not exceed the threshold
(\Cref{table:split}).
This experiment illustrates that \algname is capable of constructing the compacted graph for a dataset of five human genomes under memory restrictions
commensurate with a low-end laptop.

Our last experiment assessed the effects of the input size and structure (number of junctions and number of distinct \kmers) 
on running time and memory consumption (\Cref{plot:structure}).
As expected from the theoretical analysis, the running time depends both on the input size and structure, while memory consumption depends only on structure.
For example, consider the dataset from~\citet{baier15}, 
which has highly similar genomes.
As a result, the number of distinct \kmers and junctions is nearly constant even as the number of genomes increases.
This dataset has the lowest running time, and the amount of memory \algname uses does not increase with the number of genomes.
Unlike the memory usage, the running time does see a dominant effect of the input size, as the running time increases with the number of genomes for this dataset.
On the other hand, consider the primates dataset, which is more variable and contains more distinct \kmers and junctions than the simulated human dataset. 
As a result, \algname takes a longer time and has larger memory consumption.

\begin{table*}[t]
	\caption{\small The minimal number of rounds it takes for \algname to compress the graph without exceeding a given memory threshold.
In this experiment we used five simulated human genomes.
Memory quantities are in gigabytes and running times are in minutes.
It was carried out on a machine with a Intel Xeon E7-8837 processor.
We used $k = 25$ and ran the computation with eight worker threads.
In each run we used the largest possible Bloom filter size that fitted a given restriction (in our implementation the number of bits it has to be a power of two).
}
	\label{table:split}
	\small
	\begin{center}		
		\begin{tabular}{|l|r|r|r|r|}\hline
			Memory threshold & Used memory & Bloom filter size & Running time & Rounds \\ \hline
			10 & 8.62 &8.59 & 259 & 1 \\
			8 & 6.73 &4.29 & 434 & 3\\
			6 & 5.98 &4.29 & 539 & 4\\
			4 & 3.51 &2.14 & 665 & 6\\
			\hline
		\end{tabular}
	\end{center}
\end{table*}

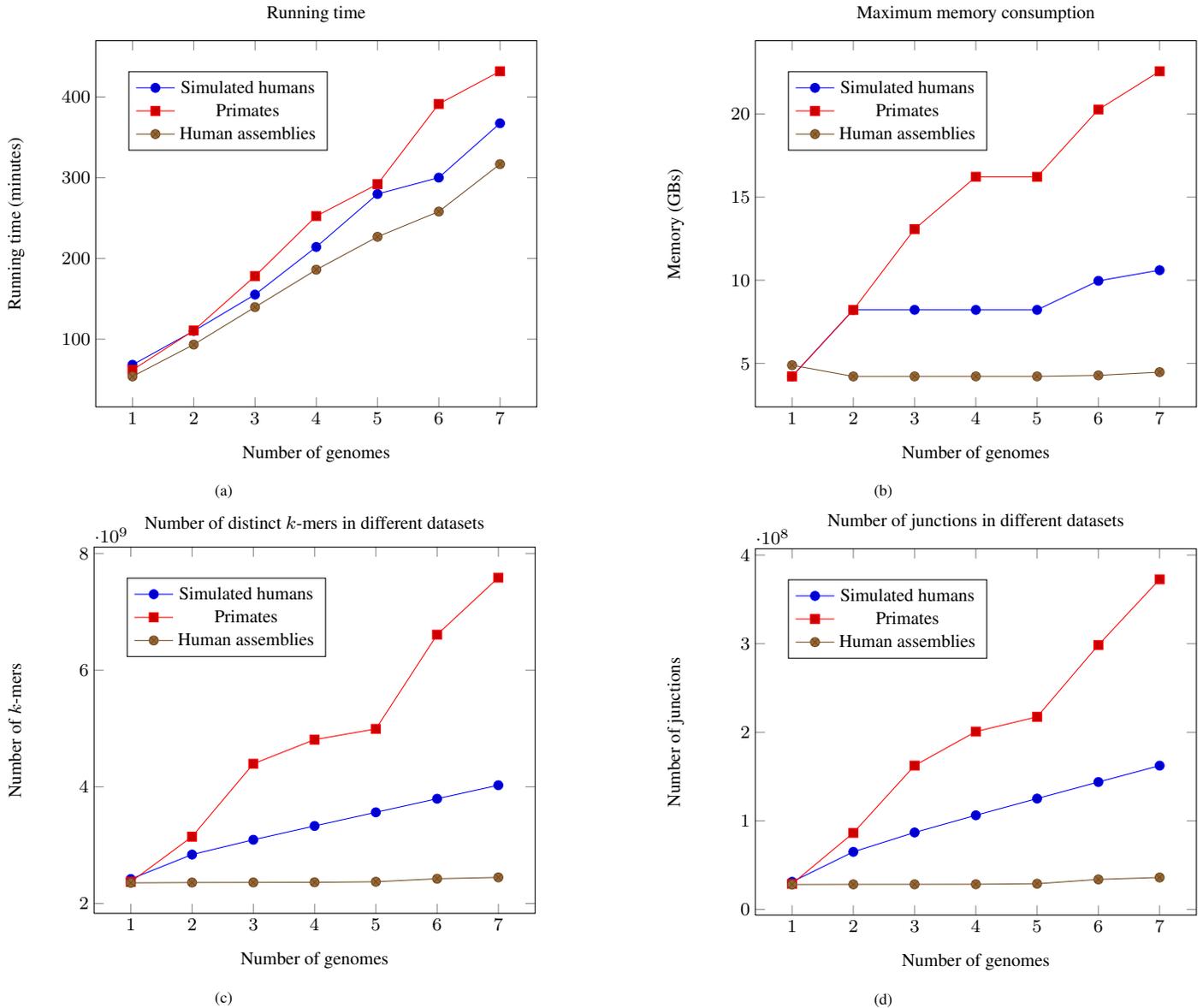
\begin{figure*}[t!]
	\begin{subfigure}{.4\textwidth}
		\begin{tikzpicture}
		\begin{axis}
		[
			title=Running time,
			xlabel=Number of genomes,
			ylabel=Running time (minutes),
			legend entries={Simulated humans, Primates, Human assemblies},
			legend style={at={(0.3,0.7)},anchor=south},
		]
		\addplot table [x=N, y=adjruntime, col sep=comma] {results_k25_hum_scale.csv};
		\addplot table [x=N, y=adjruntime, col sep=comma] {results_k25_prim_scale.csv};
		\addplot table [x=N, y=adjruntime, col sep=comma] {results_k25_refhum_scale.csv};
		\end{axis}
		\end{tikzpicture}
		\caption{}
		\label{plot:structure:runtime}
	\end{subfigure}
	\hfill
	\begin{subfigure}{.4\textwidth}
		\begin{tikzpicture}
		\begin{axis}
		[
			title=Maximum memory consumption,
			xlabel=Number of genomes,
			ylabel=Memory (GBs),
			legend entries={Simulated humans, Primates, Human assemblies},
			legend style={at={(0.3,0.7)},anchor=south},
		]
		\addplot table [x=N, y=adjmemory, col sep=comma] {results_k25_hum_scale.csv};
		\addplot table [x=N, y=adjmemory, col sep=comma] {results_k25_prim_scale.csv};
		\addplot table [x=N, y=adjmemory, col sep=comma] {results_k25_refhum_scale.csv};
		\end{axis}
		\end{tikzpicture}
		\caption{}
		\label{plot:structure:memory}
	\end{subfigure}	
	\begin{subfigure}{.4\textwidth}
		\begin{tikzpicture}
		\begin{axis}
		[
			title=Number of distinct \kmers in different datasets,
			xlabel=Number of genomes,
			ylabel=Number of \kmers,
			legend entries={Simulated humans, Primates, Human assemblies},
			legend style={at={(0.3,0.7)},anchor=south},
		]
		\addplot table [x=N, y=kmers, col sep=comma] {results_k25_hum_scale.csv};
		\addplot table [x=N, y=kmers, col sep=comma] {results_k25_prim_scale.csv};
		\addplot table [x=N, y=kmers, col sep=comma] {results_k25_refhum_scale.csv};
		\end{axis}
		\end{tikzpicture}
		\caption{}
		\label{plot:structure:kmers}
	\end{subfigure}
	\hfill
	\begin{subfigure}{.4\textwidth}
		\begin{tikzpicture}
		\begin{axis}
		[
			title=Number of junctions in different datasets,
			xlabel=Number of genomes,
			ylabel=Number of junctions,
			legend entries={Simulated humans, Primates, Human assemblies},
			legend style={at={(0.3,0.7)},anchor=south},
		]
		\addplot table [x=N, y=junctions, col sep=comma] {results_k25_hum_scale.csv};
		\addplot table [x=N, y=junctions, col sep=comma] {results_k25_prim_scale.csv};
		\addplot table [x=N, y=junctions, col sep=comma] {results_k25_refhum_scale.csv};
		\end{axis}
		\end{tikzpicture}
		\caption{}
		\label{plot:structure:junctions}
	\end{subfigure}	
	\caption{\small Effects of the input length and structure on the memory and running time.
Here we varied the number of input genomes from one to seven and recorded the running time (\subref{plot:structure:runtime}) and memory usage (\subref{plot:structure:memory}).
We also calculated the number of distinct \kmers (\subref{plot:structure:kmers}) and junctions (\subref{plot:structure:junctions}) in the input to illustrate their effect on the algorithm's performance.
We used three datasets: simulated humans, primates, and 7 human assemblies from~\citet{baier15}.
The experiment was performed on a machine with a Intel Xeon E7-8837 processor.
We used $k = 25$ and ran the computation with eight worker threads and a single round.
For each run we used the optimal Bloom filter size, i.e. the filter size that minimizes the maximum memory consumption.
The number of distinct \kmers was computed using the KMC2 \kmer counter \cite{deorowicz2015kmc}.
In our implementation, the number of bits in the Bloom filter has to be a power of two, which leads to the non-smooth growth of the memory curve
in (\subref{plot:structure:memory}).
\label{plot:structure}}
\end{figure*}

\section{Conclusion}
In this paper we gave an efficient algorithm for constructing the compacted de Bruijn graph for a collection of 
complete genomic sequences.
It is based on identifying the positions of the genome which correspond to vertices of the compacted graph.
\algname works by narrowing down the set of candidates using a probabilistic data structure,
in order to make the deterministic memory-intensive approach feasible.
We note that the effectiveness of the algorithm relies on having whole genome sequences,
making it inapplicable to the case when genomes are represented as shorts read fragments.
Parallel speedup of the second pass of~\Cref{alg:twopass} is an important direction of the future work that we are going to pursue.

A critical parameter of the \algname is the size of the Bloom filter ($b$).
We recommend the user to set $b$ to be the maximum memory they wish to allocate to the algorithm. 
If the memory usage then exceeds $b$ (which would happen due to the size of the hash table), 
then the number of rounds should be increased until the memory usage falls below $b$.
In future work, we plan to implement an algorithm to automatically select a value of $b$ 
that would minimize the maximum memory used by the algorithm.
We also plan to automate the choice of the number of rounds, given a desired memory limit.

The algorithm can also be used to construct a partially compacted graph by omitting the second pass of~\Cref{alg:twopass}.
A partially compacted graph is one where some, but not necessarily all, of the non-branching paths have been compacted.
Partially compacted graphs are faster to construct and can be useful in applications 
when the size of the graph is not critical or full compaction takes too much resources.

\algname makes significant progress in 
extending the number and size of genomes from which a compacted de Bruijn graph can be constructed.
We believe that this progress will enable novel biological analyses of mammalian-sized genomes.
For example, de Bruijn graphs can now be applied to construct synteny blocks for closely related mammalian species, similar to how they were applied to bacterial genomes \citep{minkin2013sibelia,pham2010drimm}.
\algname can also be useful in other applications, 
such as the representation of multiple reference genomes 
or variants between genomes.

\section*{Acknowledgements}

We would like to thank Daniel Lemire for modifying his 
hash function library~\citep{lemire2010recursive} for the purpose of our algorithm.

\section*{Funding}
This work has been supported in part by NSF awards DBI-1356529, CCF-1439057, IIS-1453527, and IIS-1421908 to PM.
\vspace*{-12pt}

%
%

\bibliographystyle{natbib}
\bibliography{bibliography}
\end{document}